\documentclass{ifacconf}
\usepackage{graphicx}      
\usepackage{natbib}        
\usepackage{amsmath,amssymb,amsfonts}
\usepackage{xcolor}
\usepackage{arydshln} 
\usepackage{enumitem}

\usepackage{url}

\newtheorem{assumption}[thm]{Assumption}

\newcommand{\ml}{}

\begin{document}
\begin{frontmatter}

\title{A Common Lyapunov Matrix Approach to the Exponential Stability of Augmented Primal-Dual Gradient Flow as LPV Systems\thanksref{footnoteinfo}} 

\thanks[footnoteinfo]{This work was supported by JSPS KAKENHI under grant numbers 24K23864 and 26K17401; National Natural Science Foundation of China under Grant 62173155 and 52188102.
This work was also supported by Japan Science and Technology Agency (JST) as part of Adopting Sustainable Partnerships for Innovative Research Ecosystem (ASPIRE), Grant Number JPMJAP2402.}

\author[First]{Mengmou Li} 
\author[Second]{Lijun Zhu} 
\author[First]{Masaaki Nagahara}

\address[First]{Hiroshima University (e-mail: mmli.research@gmail.com; nagahara@ieee.org).}
\address[Second]{Huazhong University of Science and Technology (e-mail: ljzhu@hust.edu.cn)}

\begin{abstract}
We show that a common Lyapunov matrix exists for the convex combination of two Hurwitz matrices if and only if the intersection of the set of strict Lyapunov matrices for one matrix and the set of non-strict Lyapunov matrices for the other is nonempty.
This simple relaxation is useful for the convergence analysis of the augmented primal-dual gradient flow for constrained optimization problems with affine inequality constraints, which can be viewed as a polytopic linear parameter-varying (LPV) system driven by the active-constraint selector. 
Under a relaxed strong convexity condition, exponential convergence is proved for the LPV system. The analysis can further be extended to the integral quadratic constraints (IQCs) framework for LPV systems to facilitate numerical search of the convergence rate.
\end{abstract}


\end{frontmatter}

\section{Introduction}
The study of exponential convergence rates of optimization algorithms, for both constrained and unconstrained problems, has attracted significant research attention in recent years (see, e.g., \cite{lessard2016analysis}, \cite{taylor2018lyapunov}, \cite{li2021exponentially}, \cite{li2025convergence}, \cite{li2025first}). 
Among them, the primal–dual dynamics have received particular interest owing to their simple implementation yet analytically challenging behavior (see, e.g., \cite{dhingra2018proximal,cortes2019distributed, qu2019exponential,tang2020semi,wang2021exponential,van2023automated,ozaslan2023tight,li2025exponential}.)
Specifically, \cite{qu2019exponential} showed exponential convergence for strongly convex functions with affine equality and inequality constraints. 
For general nonlinear inequality constraints, semi-global exponential stability has been established in \cite{tang2020semi}.
For equality constraints, the requirement of global strong convexity can be relaxed to strong convexity restricted to the subspace defined by the equality matrix under the augmented Lagrangian framework (\cite{li2025exponential}).

Similarly to their use in analyzing unconstrained algorithms, the robust control framework of integral quadratic constraints (IQCs) (\cite{megretski1997system,lessard2016analysis}) can also be employed to study constrained algorithms (\cite{dhingra2018proximal,li2025exponential}) by formulating the dynamics as a Lur’e system, that is, a feedback interconnection between a linear time-invariant (LTI) system and a nonlinearity.
However, this framework has not been applied to the primal–dual dynamics with inequality constraints yet, owing to their non-smooth nature that cannot be represented by an LTI system, even though such primal–dual formulations arise in many applications such as network congestion control and frequency regulation in power systems (\cite{low1999optimization,li2023distributed}).

In the standard Lur’e representation, nonlinearities are typically collected together as the ``troublesome'' component. Nevertheless, certain systems involve benign nonlinear terms that can instead be interpreted as switched linear systems or linear parameter-varying (LPV) systems admitting a common quadratic Lyapunov function (CQLF).
Recent studies have extended the IQC framework to handle time-varying but non-vanishing stepsizes by modeling such algorithms as LPV systems with nonlinearities (see \cite{jakob2025linear} and references therein).
However, it is worth mentioning that the assumption of a common constant Lyapunov function \ml{could be} restrictive (\cite{chesi2007robust}). 
Moreover, as linear matrix inequalities (LMIs) are typically solved numerically and thus only provide numerical certificates, it is desirable to establish
\ml{analytical results that guarantee} 
the existence of a common Lyapunov matrix \textit{a priori} before applying this framework. This work fulfills this role by providing an analytical proof for the augmented primal–dual gradient flow under affine inequality constraints, within the LPV framework in connection with the gradient nonlinearity.

In this work, we focus on the primal-dual gradient flow derived from the augmented Lagrangian (\cite{bertsekas2014constrained}). The vector field is continuous but not differentiable everywhere (\cite{qu2019exponential}). 
Although smooth gradient flows have been proposed for handling inequality constraints (see, e.g., \cite{li2018generalized}), they introduce additional nonlinearities that complicate the analysis.
Instead of designing specific Lyapunov functions for primal-dual dynamics, we construct a common Lyapunov matrix for the convex hull of a group of Hurwitz matrices corresponding to the system modes associated with the active and inactive statuses of the inequality constraints.
Compared with \cite{qu2019exponential}, our work relaxes the global strong convexity for the affine inequality constrained problems, in a manner analogous to the equality-constrained case (\cite{li2025exponential}).
The relaxation induced by inequality constraints is inherently more restrictive, since an inequality constraint may not always be active.  When inactive, the problem effectively reduces to an unconstrained one \ml{with respect to that constraint}.
Furthermore, the proposed framework can be extended to the exponential stability analysis of distributed optimization algorithms under various affine constraints, provided that the corresponding linear systems admit a common Lyapunov matrix for all Hurwitz vertices.

The remainder of the paper is organized as follows.
Section~\ref{sec: pre and main} presents the preliminaries and the main results on the common Lyapunov matrix.
In Section~\ref{sec: app to con opt}, the proposed framework is applied to optimization problems with affine inequality constraints.
A numerical example is provided in Section~\ref{sec: example}, and concluding remarks are given in Section~\ref{sec: conclusion}.

\section{Preliminaries \& Main Results}\label{sec: pre and main}
\subsection{Notation}
$\mathbb{R}$ and $\mathbb{C}$ represent the sets of real and complex numbers, respectively. $\mathbb{R}^{n}$ and $\mathbb{R}^{n \times m}$ denote the sets of $n$-dimensional real vectors and $n \times m$ real matrices, respectively.
$\mathbb{S}_{++}^{n}$ denotes the set of $n \times n $ positive definite matrices.
The symbols $I_{n}$, and $\mathbf{0}_{n \times m}$ represent the $n\times n$ identity matrix and the $n \times m$ zero matrix, respectively, while their subscripts may be omitted.  The symbol $\mathbf{1}$ denotes the $n$-dimensional vector whose entries are all equal to one.
The operators $\text{diag}(v)$ and $\text{blkdiag}(A_1, \ldots, A_k)$ denote the diagonal matrix from the vector $v$ and block-diagonal matrix from the square matrices $A_1, \ldots, A_k$, respectively.
The Kronecker product is denoted by $\otimes$, and the Hadamard (entrywise) product by $\circ$.  The Euclidean norm is denoted by $\|\cdot\|$.  
Given symmetric matrices $A$ and $B$, $A \prec B$ means that $B- A$ is positive definite, and $\succ$, $\preceq$, $\succeq$ are defined \textit{mutatis mutandis}. 
The vector $e_j \in \mathbb{R}^{m}$ represents the $j$-th standard basis vector.
The kernel (nullspace) of a matrix $A$ is defined by $\text{ker} (A) = \left\{ x: A x = 0\right\}$.
The convex hull of a set of $N$ points $\left\{ p_{(1)}, \ldots, p_{(N)}\right\}$ is defined as
$\mathrm{conv} \left(  p_{(1)}, \ldots, p_{(N)} \right) = \left\{
\sum_{i=1}^{N} \theta_{(i)} p_{(i)} :
\theta_{(i)} \ge 0,~ \forall i,~
\sum_{i=1}^{N} \theta_{(i)} = 1
\right\}.$

\subsection{Common Lyapunov matrix}
We define three sets of Lyapunov matrices associated with a given matrix $A \in \mathbb{R}^{n \times n}$:
\begin{itemize}
    \item The set of \textbf{strict Lyapunov matrices}:
    \[
    \mathcal{P}_A := \left\{ P \in \mathbb{S}_{++}^n : PA + A^\top P \prec 0 \right\}
    \]
    \item The set of \textbf{Lyapunov matrices with margin \( \varepsilon > 0 \)}:
    \[
    \mathcal{P}_A^\varepsilon := \left\{ P \in \mathbb{S}_{++}^n : PA + A^\top P \preceq -\varepsilon I \right\}
    \]
    \item The set of \textbf{non-strict Lyapunov matrices}:
    \[
    \bar{\mathcal{P}}_A := \left\{ P \in \mathbb{S}_{++}^n : PA + A^\top P \preceq 0 \right\}
    \]
\end{itemize}
Clearly, $\mathcal{P}_A^\varepsilon \subset \mathcal{P}_A \subset \bar{\mathcal{P}}_A$. 

Consider the continuous-time state space system
\begin{align}\label{eq: LPV system}
\dot{x} (t) = A(\theta (t)) x (t)
\end{align}
where $\theta (t) \in \Theta \subset \mathbb{R}^{q}$ is a parameter that can be time varying, 
$\Theta := \mathrm{conv} \left( \theta_{(1)}, \ldots, \theta_{(N)}\right)$ is a polytope of $N$ vertices, and $A (\theta (t))$ depends affinely on $\theta(t)$.

Given a list of Hurwitz matrices $A_{(i)} \in \mathbb{R}^{n \times n}$, $i \in \mathcal{I} : =  \left\{ 1, \ldots, N\right\}$,  we say that there exists a \textit{common Lyapunov matrix} (CLM)  if there exists $P \succ 0$, such that
\begin{align}\label{eq: definition of CLM}
P A_{(i)} + A_{(i)}^{\top}  P  \prec 0, ~ \forall i \in \mathcal{I}.
\end{align}
In other words, $ \bigcap_{i \in \mathcal{I}} \mathcal{P}_{A_{(i)}} \neq \emptyset$. If $A_{(i)} = A (\theta_{(i)})$, $i \in \mathcal{I}$, which correspond to the $N$ vertices of polytope $\Theta$ in \eqref{eq: LPV system}, then $A (\theta (t))$ lies in the convex hull $\mathrm{conv}\left(A_{(1)},\dots,A_{(N)}\right)$, and there exists a \textit{common quadratic Lyapunov function} (CQLF) $V = x^\top P x$ for system \eqref{eq: LPV system}.

Condition \eqref{eq: definition of CLM} is a set of linear matrix inequalities (LMIs) which can be solved efficiently. However, it may be difficult to verify analytically when the scale of the problem is large.  Instead, we could relax it and obtain an easier-to-verify condition that is also sufficient and necessary for the existence of a CLM.
We state our main result in terms of two Hurwitz matrices.
\begin{thm}\label{thm: common Lyapunov matrix for Hurwitz matrices}
Let \( A_{(1)}, A_{(2)} \in \mathbb{R}^{n \times n} \) be two Hurwitz matrices.  
For \( \theta \in [0,1] \), define \( A_\theta := \theta A_{(1)} + (1 - \theta) A_{(2)} \).  
Then the following are equivalent:
\begin{enumerate}[label=\arabic*)]
    \item \label{itm:A} There exists a common Lyapunov matrix $P \succ 0$ and a constant $\varepsilon > 0$ independent of $\theta$, such that
    \begin{align*}
        P \in \mathcal{P}_{A_{\theta}}^\varepsilon, \text{ for all } \theta \in [0,1],
    \end{align*}
    that is, $P \in \bigcap_{\theta \in [0,1]} \mathcal{P}_{A_\theta}^\varepsilon$.
    \item \label{itm:B} The intersection of $\mathcal{P}_{A_{(1)}}$ and $\bar{\mathcal{P}}_{A_{(2)}} $ is nonempty, that is,
    \begin{align}\label{eq:intersection}
        \mathcal{P}_{A_{(1)}} \cap \bar{\mathcal{P}}_{A_{(2)}} \neq \emptyset, 
    \end{align}
\end{enumerate}
\end{thm}
\begin{pf}
\newline
\ref{itm:A} $\implies$ \ref{itm:B}:  
This is obvious, since $P \in \left( \mathcal{P}_{A_{(1)}} \cap \mathcal{P}_{A_{(2)}} \right) \subset \left( \mathcal{P}_{A_{(1)}} \cap \bar{\mathcal{P}}_{A_{(2)}} \right)$.\\
\ref{itm:B} $\implies$ \ref{itm:A}: 
Assume that \eqref{eq:intersection} holds, that is, there exists \( P_{(1)} \in \mathcal{P}_{A_{(1)}} \cap \bar{\mathcal{P}}_{A_{(2)}} \) such that
\begin{align*}
\quad P_{(1)} A_{(1)} + A_{(1)}^{\top} P_{(1)} \prec 0, \quad P_{(1)} A_{(2)} + A_{(2)}^\top P_{(1)} \preceq 0.
\end{align*}
It is clear that there exists $\epsilon > 0$ such that $P_{(1)} \in \mathcal{P}_{A_{(1)}}^{\epsilon}$.
Next, pick a bounded \( P_{(2)} \in \mathcal{P}_{A_{(2)}} \), and define
$P_{\alpha} := \alpha P_{(1)} + P_{(2)}$, $\alpha > 0$.
Clearly, \( P_\alpha \succ 0 \) for any \( \alpha > 0 \).
We have
\begin{align*}
& P_\alpha A_\theta + A_\theta^\top P_\alpha\\
= & \left(\alpha P_{(1)} + P_{(2)}\right) \left(\theta A_{(1)} + (1 - \theta) A_{(2)} \right)\\
& \quad + \left(\theta A_{(1)}^\top + (1 - \theta) A_{(2)}^\top)(\alpha P_{(1)} + P_{(2)}\right) \\
= & \alpha \theta \left(P_{(1)} A_{(1)} + A_{(1)}^\top P_{(1)} \right) + \theta \left(P_{(2)} A_{(1)} + A_{(1)}^\top P_{(2)} \right) \\
& + (1 - \theta) \left(P_{(2)} A_{(2)} + A_{(2)}^\top P_{(2)} \right)\\
&  + \alpha (1 - \theta) \left(P_{(1)} A_{(2)} + A_{(2)}^\top P_{(1)}\right) \\
\preceq & \alpha \theta \left(P_{(1)} A_{(1)} + A_{(1)}^\top P_{(1)}\right) +  \theta \left(P_{(2)} A_{(1)} + A_{(1)}^\top P_{(2)}\right) \\
& + (1 - \theta) \left(P_{(2)} A_{(2)} + A_{(2)}^\top P_{(2)}\right)
\end{align*}
where the inequality follows from $P_ {(1)} \in \bar{\mathcal{P}}_{A_ {(2)} }$.
As $\left\| P_ {(2)}  A_ {(1)}  + A_ {(1)} ^\top P_ {(2)} \right\|$ is bounded, there exists a sufficiently large $\alpha$ independent of $\theta$ such that 
$$\alpha \left(P_ {(1)}  A_ {(1)}  + A_ {(1)} ^\top P_ {(1)} \right) +  \left(P_ {(2)}  A_ {(1)}  + A_ {(1)} ^\top P_ {(2)} \right) \prec - \delta I_{n}$$
for some $\delta > 0$.
Therefore, the convex combination of the previous terms with $\left(P_ {(2)}  A_ {(2)}  + A_ {(2)} ^\top P_ {(2)} \right)$ yields that $P_\alpha A_\theta + A_\theta^\top P_\alpha \preceq - \varepsilon I$ with a margin $\varepsilon > 0$ uniformly over \( \theta \in [0,1] \), that is, $P_{\alpha} \in \bigcap_{\theta \in [0,1]} \mathcal{P}_{A_\theta}^\varepsilon$, for some $\varepsilon > 0$.
\end{pf}
The above result can be extended from two matrices to the convex hull of multiple matrices by repeatedly applying the above theorem to the obtained $A_{\theta}$ and a new vertex matrix, thereby covering the entire polytope.
\begin{cor}\label{cor: convex hull CLM}
Let $\left\{A_{(i)}\right\}$ be a collection of Hurwitz matrices.
Suppose there exist $P_{(k)} \succ 0$ and $\varepsilon_{(k)}>0$ such that
\begin{align*}
P_{(k)} \in \mathcal{P}_{A}^{\varepsilon_{(k)}},
\quad \text{for all } A\in \mathrm{conv}\left\{A_{(1)},\dots,A_{(k)}\right\}.
\end{align*}
Assume further that for a new vertex $A_{(k+1)}$ the non-strict inequality holds, that is,
$P_{(k)} \in \bar{\mathcal{P}}_{A_{(k+1)}}$.
Then there exist $P_{(k+1)} \succ 0$ and $\varepsilon_{(k+1)}>0$ such that
\begin{align*}
P_{(k+1)} \in \mathcal{P}_{A}^{\varepsilon_{(k+1)} }
\quad \text{for all } A\in \mathrm{conv}\left\{A_{(1)},\dots,A_{(k+1)}\right\}.
\end{align*}
\end{cor}

The following lemma was shown in \cite{li2025exponential}, and will be directly used here.
\begin{lem}
Consider the real matrix
\begin{align}\label{eq: sum of skew and block diagonal}
A = \begin{bmatrix}
	 	- F & -T \\
	 	T^{\top} & \mathbf{0}
\end{bmatrix}
\end{align}
where $F \in \mathbb{R}^{n \times n}$, $T \in \mathbb{R}^{n \times m}$, and $m \leq n$.
If $F \succ 0$, and $T^{\top}$ is of full row rank, then $A$ is Hurwitz.
\label{lem:Hurwitz}
\end{lem}
Consequently, there exists $P \in \mathcal{P}_{A}$, for $A$ given in Lemma~\ref{lem:Hurwitz}.

A natural extension of Lemma~\ref{lem:Hurwitz} is to consider a parameter-dependent positive definite matrix $F(x)$ using the S-procedure (\cite{polik2007survey}).
\begin{cor}\label{cor: varying parameter Hurwitz}
Consider the real matrix
\begin{align}\label{eq: existence of }
A(x) = \begin{bmatrix}
	 	- F(x) & -T \\
	 	T^{\top} & \mathbf{0}
\end{bmatrix}
\end{align}
where $T \in \mathbb{R}^{n \times m}$, $m \leq n$, $F (x) \in \mathbb{R}^{n \times n}$ and satisfies $\mu I_n \preceq F(x) \preceq \ell I_n$, with constants $\ell \geq\mu > 0$. 
Then, there exists a common Lyapunov matrix  $P \in \mathcal{P}_{A (x)}$ for any $x$.
\end{cor}

\section{Application to Inequality Constrained Optimization}\label{sec: app to con opt}
Consider the inequality-constrained optimization problem 
\begin{equation}\label{eq:problem}
\begin{aligned}
	&\min_{x \in \mathbb{R}^{n}} f(x) \\
		& \text{subject to } T^{\top} x \leq b
\end{aligned}
\end{equation}
where the objective function $f(x): \mathbb{R}^{n} \rightarrow \mathbb{R}$ is continuously differentiable, $T \in \mathbb{R}^{n \times m}$ with $m \leq n$, and $b\in \mathbb{R}^{m}$. 
We denote $T = \left[T_1, \ldots, T_m \right]$,  with $T_i \in \mathbb{R}^{n}$, $b = \left[b_1, \ldots, b_m \right]^{\top}$ and the inequality is to be understood elementwise.
The following assumptions are commonly considered.
\begin{assumption}\label{assumption general convex and smooth}
$f(x)$ is convex and $\ell$-Lipschitz smooth, that is, there exists a constant $\ell > 0$, such that for any $x,~x' \in \mathbb{R}^{n}$,
\begin{align*}
   0\leq  \left( \nabla f(x) - \nabla f(x') \right)^\top (x - x') \leq  \ell \| x - x' \|^2.
\end{align*}
\end{assumption}

\begin{assumption}\label{assumption constraint qualification}
	The matrix $T^{\top}$ is of full row rank and satisfies $ \kappa_1 I_m \preceq T^{\top} T  \preceq \kappa_2 I_m$ for some $\kappa_2 \geq \kappa_1 > 0$.
\end{assumption}
When all constraints are active,  the above condition is known as the linear independence constraint qualification (LICQ) (\cite{nocedal1999numerical}).
It has been shown in \cite{qu2019exponential} that a primal-dual gradient algorithm for \eqref{eq:problem} achieves global exponential convergence if $f(x)$ is strongly convex and the strong convexity can be relaxed to the subspace within the kernel of $T^\top$ (\cite{li2025exponential}). However, for inequality constraints, the relaxation induced by the constraints is more restricted, since an inequality constraint may not always be active, and when inactive, the problem effectively reduces to an unconstrained one. The corresponding conditions are stated below.
\begin{assumption}\label{assumption strongly convex region}
$f(x)$ is $\mu$-strongly convex in $\text{ker} (T^\top)$, and whenever an inequality constraint is inactive, that is,
\begin{align*}
    & \left( \nabla f(x) - \nabla f(x') \right)^\top (x - x') \geq \mu \| x - x'\|^2, \\
    &\text{for all } x, x' \in \text{ker} (T^\top), \\
    & \text{and for all } x, x' \in \left\{ y \in \mathbb{R}^{n} : T_i^{\top} y < b_i \right\}, ~i = 1,\ldots, m.
\end{align*}
\end{assumption}
\begin{rem}
\ml{This assumption is weaker than requiring strong convexity over the whole domain (e.g., \cite{dhingra2018proximal,qu2019exponential,tang2020semi}).} The strong convexity parameters for the above subsets can be different but we use the same value here for simplicity. 
\end{rem}
Under the above assumptions, problem \eqref{eq:problem} is convex. The existence of its solution is guaranteed by the constraint qualification in Assumption~\ref{assumption constraint qualification}, while its uniqueness follows from the strong convexity of the objective function within the feasible set, as stated in Assumption~\ref{assumption strongly convex region}.
\begin{lem}\label{lem: unique solution}
Under Assumptions \ref{assumption general convex and smooth}, \ref{assumption constraint qualification}, and \ref{assumption strongly convex region}, problem \eqref{eq:problem} has a unique solution.
\end{lem}

\subsection{Augmented primal-dual gradient flow}
The inequality in \eqref{eq:problem} can be replaced by equality 
$T^\top x - b + z \circ z = 0$
where $z \in \mathbb{R}^{m}$ is a slack variable and $\circ$ represents the elementwise product.
Applying the augmented Lagrange multiplier method gives (\cite{bertsekas2014constrained})
\begin{align*}
\mathcal{L} (x, \lambda, z) = &  f(x) + \sum_{j = 1}^{m} \lambda_j \left( T_j^\top x - b_j + z_j^2 \right)\\
& \qquad \quad + \frac{\rho}{2}   \sum_{j = 1}^{m}  \left(  T_j ^\top x - b_j + z_j^2 \right)^2.
\end{align*}
Note that minimization with respect to $z$ can be solved explicitly,  then the derived augmented Lagrangian for solving \eqref{eq:problem} is formulated as (\cite{qu2019exponential}),
\begin{align}
    L (x, \lambda) = \min_{z} \mathcal{L} (x, \lambda, z)  = f(x) + \sum_{j = 1}^{m} H_{\rho} (T_j^{\top} x - b_j , \lambda_j)
\end{align}
where $\rho > 0$ is a free parameter and $H_{\rho}$ is given by
\begin{align*}
    & H_{\rho} (T_j^{\top} x - b_j, \lambda_j) = \\
    & 
    \begin{cases}
        (T_j^{\top} x - b_j) \lambda_j + \frac{\rho}{2} (T_j^{\top} x - b_j)^{2}, & \rho \left(T_j^{\top} x - b_j \right) + \lambda_j \geq 0\\
        -\frac{\lambda_j^2 }{2 \rho}, & \text{otherwise}
    \end{cases}
\end{align*}
The primal-descent dual-ascent gradient dynamics for $L(x, \lambda)$ is then given by 
\begin{equation}\label{eq: augmented primal-dual algorithm}
    \begin{aligned}
        & \dot{x} = -\nabla f(x) - \sum_{j = 1}^{m} \max\left\{ \rho \left( T_j^{\top} x - b_j \right) + \lambda_j , 0 \right\} T_j \\
        & \dot{\lambda} = \sum_{j = 1}^{m} \frac{\max\left\{ \rho \left( T_j^{\top} x - b_j \right) + \lambda_j , 0 \right\} - \lambda_j}{\rho} e_j
\end{aligned}
\end{equation}
where $e_j \in \mathbb{R}^{m}$ is the $j$-th standard basis vector.

Let $x^*$ be the unique optimal solution of \eqref{eq:problem}. 
Under Assumption~\ref{assumption general convex and smooth}, $\nabla f (x)$ is Lipschitz continuous, it is differentiable almost everywhere by Rademacher's theorem (\cite{clarke1990optimization}), and there exists $0 \preceq F(x) \preceq \ell I_n$ such that $\nabla f (x) - \nabla f(x^*) = F(x) (x - x^*)$ (\cite{qu2019exponential}). When $f$ is strongly convex, $F(x) \succeq \mu I_n$.

\subsection{Exponential convergence}
\begin{thm}\label{thm exponential convergence}
Under Assumptions~\ref{assumption general convex and smooth}, \ref{assumption constraint qualification}, and \ref{assumption strongly convex region}, the variable $x(t)$ in the augmented primal-dual gradient flow \eqref{eq: augmented primal-dual algorithm} exponentially converges to the optimal solution of problem \eqref{eq:problem}.
\end{thm}
\begin{pf}
Let $( x^*, \lambda^*)$ be the equilibrium point of \eqref{eq: augmented primal-dual algorithm},  and denote $\tilde{x} = x - x^*$, $\tilde{\lambda} = \lambda - \lambda^*$, then the error dynamics can be written as
\begin{align}\label{eq: error system compact}
    \begin{bmatrix}
    \dot{\tilde{x}}\\ \dot{\tilde{\lambda}}    
    \end{bmatrix}
    = & 
    \underbrace{
    \begin{bmatrix}
    -F(x) - \rho T \Gamma (x, \lambda) T^{\top} & - T \Gamma(x, \lambda)\\
    \Gamma(x, \lambda) T^{\top} &  \frac{1}{\rho} ( \Gamma(x, \lambda) - I)
    \end{bmatrix}
    }_{G( F(x), \Gamma (x, \lambda))}
    \begin{bmatrix}
    {\tilde{x}}\\{\tilde{\lambda}}    
    \end{bmatrix}
\end{align}
where $\Gamma (x, \lambda) \in \mathbb{R}^{m \times m}$ is a diagonal matrix function depending on $x$, $\lambda$, and satisfies $0 \preceq \Gamma (x, \lambda)  \preceq I_{m}$.\\
Let us define
$H (F(x), \theta (t) ) := G(F(x), \text{diag} (\theta (t) ))$, where $\theta (t) \in [0, 1] \times \ldots \times [0, 1] \subset \mathbb{R}^{m}$, that is,
    \begin{align}\label{eq: matrix to be analyzed}
        H (F(x), \theta)= 
        \begin{bmatrix}
            - F(x) - \rho T \text{diag} (\theta ) T^{\top} & - T \text{diag} (\theta )\\
            \text{diag} (\theta ) T^{\top} &  \frac{1}{\rho} ( \text{diag} (\theta ) - I_m)
        \end{bmatrix}.
    \end{align}
The proof is reduced to showing that $H (F(x), \theta)$ is Hurwitz and there exists a CLM $P \in \mathcal{P}_{H (\theta)}^{\varepsilon}$ with some $\varepsilon > 0$, uniformly for any $x$ and $\theta \in [0, 1]^{m}$.
Define
\begin{align*}
    & H (F(x), \mathbf{0}) = \begin{bmatrix}
    - F(x) & \mathbf{0} \\
     \mathbf{0} &  - \frac{1}{\rho} I_m
    \end{bmatrix}, \\~
& H (F(x), \mathbf{1}_m) = 
\begin{bmatrix}
    - F(x) - \rho T T^{\top} & - T \\
     T^{\top} &  \mathbf{0}
    \end{bmatrix}.
\end{align*}
Under Assumption~\ref{assumption strongly convex region}, 
$F(x) \succeq \mu I_n$ when $T^{\top} x <  b$, and $F(x) +  \rho T T^{\top}  \succeq \gamma I$, for some $\gamma > 0$, when $T^{\top} x \geq b$.
Then, the above two matrices are both Hurwitz in their respective regions, by Corollary~\ref{cor: varying parameter Hurwitz}.
Assume first that there is only one inequality constraint, that is, $m = 1$, a direct application of Theorem~\ref{thm: common Lyapunov matrix for Hurwitz matrices} with $ H (F(x), \mathbf{0})$ and $H (F(x), \mathbf{1}_m)$ proves the existence of a CLM $P_ {(1)}  = \alpha I_{n+1} + P_0$, where $P_0$ is a CLM for $H (F(x), \mathbf{1}_m) $, for all $x$.

When there are multiple constraints, we only need to apply Theorem~\ref{thm: common Lyapunov matrix for Hurwitz matrices} repeatedly. If there exists a CLM for the $N = 2^{m}$ Hurwitz matrices given at $\theta \in \left\{ 0, 1\right\}^{m}$, then the polytopic LPV system is exponentially stable.\\
Specifically, let us first consider the cases where $\theta = \mathbf{0}$ and $\theta = e_1 \in \mathbb{R}^{m}$. We have
\begin{align*}
    H (F(x), e_1 ) = 
    \left[
\begin{array}{cc:c}
- F(x) - \rho T_{1} T_{1}^{\top}  & - T_{1} & \mathbf{0} \\
T_{1}^{\top} & 0 & \mathbf{0} \\
\hdashline
\mathbf{0} & \mathbf{0} & - \frac{1}{\rho} I_{m -1}
\end{array}
\right]
\\
    H (F(x), \mathbf{0}) = 
    \left[
\begin{array}{cc:c}
- F(x)   & \mathbf{0} & \mathbf{0} \\
\mathbf{0} & -\frac{1}{\rho} & \mathbf{0} \\
\hdashline
\mathbf{0} & \mathbf{0} & - \frac{1}{\rho} I_{m -1}
\end{array}
\right].
\end{align*}
Under Assumption~\ref{assumption strongly convex region}, both $H (F(x), e_1 )$ and $H (F(x), \mathbf{0})$ are Hurwitz in $\left\{y : T_{1}^\top y \geq b \right\}$,  and $\left\{ y : T_{1}^\top y < b \right\}$ respectively, as they are in the form of a direct sum of a Hurwitz block with $- \frac{1}{\rho} I_{m -1}$.
Similarly to the previous argument, we can obtain by Theorem~\ref{thm: common Lyapunov matrix for Hurwitz matrices} that the convex combination of the two matrices given by 
\begin{align*}
    H (F(x),  \theta_1 e_1 ) = 
    \left[
\begin{array}{cc:c}
- F(x)- \rho T_{1} \theta_1 T_{1}^{\top}  & - T_{1} \theta_1 & \mathbf{0} \\
\theta_1 T_1^{\top}  & -\frac{1 - \theta_1}{\rho} & \mathbf{0} \\
\hdashline
\mathbf{0} & \mathbf{0} & - \frac{1}{\rho} I_{m -1}
\end{array}
\right],
\end{align*}
admits a CLM in the form of $P_ {(2)}  = \text{blkdiag} (P_ {(1)} , I_{m-1})$ for all $\theta_1 \in [0, 1]$,  where $P_{(1)}$ is a CLM \ml{for the first diagonal blocks of $ H (F(x), e_1 )$ and $H (F(x), \mathbf{0})$,} derived similarly to the case when $m = 1$. \ml{So, $P_ {(2)} \in {\mathcal{P}}_{H (F(x),   \theta_1 e_1)}$.}
Next, let us consider the case where $\theta = \theta_1 e_1 + e_2$, that is,
\begin{align*}
    & H (F(x),   \theta_1 e_1 + e_2 ) \\
    =  & \left[
\begin{array}{cc:c:c}
- F(x) - \rho T_{1} \theta_1 T_{1}^{\top}  & - T_{1} \theta_1 &  -T_{2} & \mathbf{0} \\
\theta_1 T_1^{\top}  & -\frac{1 - \theta_1}{\rho} & 0 & 0 \\
\hdashline
T_2^{\top} & 0 & 0 & 0 \\
\hdashline
\mathbf{0} & 0 & 0 & - \frac{1}{\rho} I_{m -2}\\
\end{array}
\right].
\end{align*}
\ml{Since $P_{(1)}$ is the CLM for the first diagonal block of the above matrix, we directly have} $P_ {(2)} \in \bar{\mathcal{P}}_{H (F(x),   \theta_1 e_1 + e_2 ) }$. We then apply Theorem~\ref{thm: common Lyapunov matrix for Hurwitz matrices} again and get that matrices in the following form admit a common Lyapunov matrix,
\begin{align*}
    & H (F(x),  \theta_1 e_1 + \theta_2 e_2 ) \\
    = & \left[
\begin{array}{cc:c}
-F(x) - \rho T_{1} \theta_1 T_{1}^{\top}  & - \left[ T_{1}  ~ T_{2} \right] \text{diag} (\theta_1, \theta_2 ) & \mathbf{0} \\
\text{diag} (\theta_1, \theta_2 )  \left[ T_{1}  ~ T_{2} \right]^{\top}  & -\frac{1}{\rho} \left( I - \text{diag} (\theta_1, \theta_2 )\right) & \mathbf{0} \\
\hdashline
\mathbf{0} & \mathbf{0} & - \frac{1}{\rho} I_{m -2}
\end{array}
\right].
\end{align*}
The remaining steps follow by induction. Hence, we conclude by Corollary~\ref{cor: convex hull CLM} that $H (F(x), \theta)$ admit a CLM uniformly for all $x$ and $\theta \in [0, 1]^{m}$. Thus, system \eqref{eq: error system compact} is exponentially stable.
\end{pf}

\subsection{A Numerical Search}
We can construct a numerical search derived from static IQC (see, e.g., \cite{hu2016exponential}) for the convergence rate of the augmented primal-dual gradient flow. For ease of discussion, we consider globally strongly convex function here.
\begin{cor}\label{cor}
Suppose that the objective function $f(x)$ is $\mu$-strongly convex globally in Theorem~\ref{thm exponential convergence}, then the variable $x(t)$ converges at a rate $\alpha > 0$ if there exists $P \succ 0$ such that
\begin{align}\label{eq: lmi search}
\begin{bmatrix}
P H (\mu I_n, \theta)  + H(\mu I_n, \theta)^\top P + 2 \alpha P & P B + \ell C^{\top} \\
B^\top P + \ell C & -2 I_n
\end{bmatrix}
\preceq 0
\end{align}
for all $\theta \in [0, 1]^{m}$, 
where $H(\cdot, \cdot)$ is defined in \eqref{eq: matrix to be analyzed}, 
$B =\begin{bmatrix}
-I_{n} \\ 0_{m \times n}
\end{bmatrix} $, and
$C = \begin{bmatrix}
I_{n} & 0_{n \times m}
\end{bmatrix}$.
\end{cor}
For every fixed $\alpha$, condition \eqref{eq: lmi search} corresponds to $2^{m}$ LMIs given by vertices $\theta \in \left\{ 0, 1 \right\}^{m}$. A bisection search can be carried out to find the maximal rate $\alpha$. The existence of $P$ to \eqref{eq: lmi search} for a sufficiently small $\alpha$ is guaranteed by Theorem \ref{thm exponential convergence}.
Numerical search using more advanced multipliers can be found in \cite{lessard2016analysis}, \cite{scherer2021convex}.

\subsection{Discussion}
Instead of designing specific Lyapunov functions for primal-dual dynamics, we focus on showing the existence of a common Lyapunov matrix related to the construction of the common Lyapunov matrix for the convex hull of a group of Hurwitz matrices which are the system states associated with the active and inactive status of inequality constraints, thus certifying the exponential convergence of the augmented primal-dual dynamics. Therefore, our results can be extended to the exponential stability analysis of distributed optimization algorithms under various affine constraints, provided that the associated linear systems admit a common Lyapunov matrix for the Hurwitz matrices at their vertices.

\section{Numerical Example}\label{sec: example}
We provide a numerical example to illustrate the effect of relaxation for global convexity under affine inequality constraints.
Consider the optimization problem \eqref{eq:problem} with $n = 2$, $m = 1$, where
$
T =
\begin{bmatrix}
1 \\ 1
\end{bmatrix},
b = 2
$
and
\begin{align*}
f(x) = (\max \left\{ 0, 3 -  x_1 - x_2  \right\} )^2 + \frac{1}{2} x_2^2.
\end{align*}
It can be verified that necessary assumptions in this work are satisfied. In particular, the objective function $f(x)$ is continuously differentiable but not twice differentiable. It is globally convex, \ml{$\frac12$}-strongly convex in $\text{ker} (T^\top) = \left\{ \begin{bmatrix} x_1 \\ x_2 \end{bmatrix}: x_1 = -x_2\right\}$ and $0.438$-strongly convex on the region $\left\{ \begin{bmatrix}
x_1 \\ x_2
\end{bmatrix} : x_1 + x_2 \leq 3 \right\}$, respectively, but not globally strongly convex, which still satisfies Assumption~\ref{assumption strongly convex region}. 
The affine inequality constraint is active at the optimal solution given by
$x^* = \begin{bmatrix}
2 \\ 0
\end{bmatrix}.
$
We run algorithm \eqref{eq: augmented primal-dual algorithm} with $\rho = 1$, initial condition $x_1 (0)= x_{2} (0) = 10$, and randomly generated $\lambda (0) \in [0, 1]$. Note that $f(x)$ is not strongly convex at $x(0)$, but the inequality-constrained dynamics ensure exponential convergence toward the feasible set.
The trajectories of $x$ and the error $\| x - x^*\|$ over time are shown in Fig.~\ref{fig:1} and Fig.~\ref{fig:2}, respectively.
\begin{figure}[htbp]
\center
\includegraphics[width = 1\linewidth]{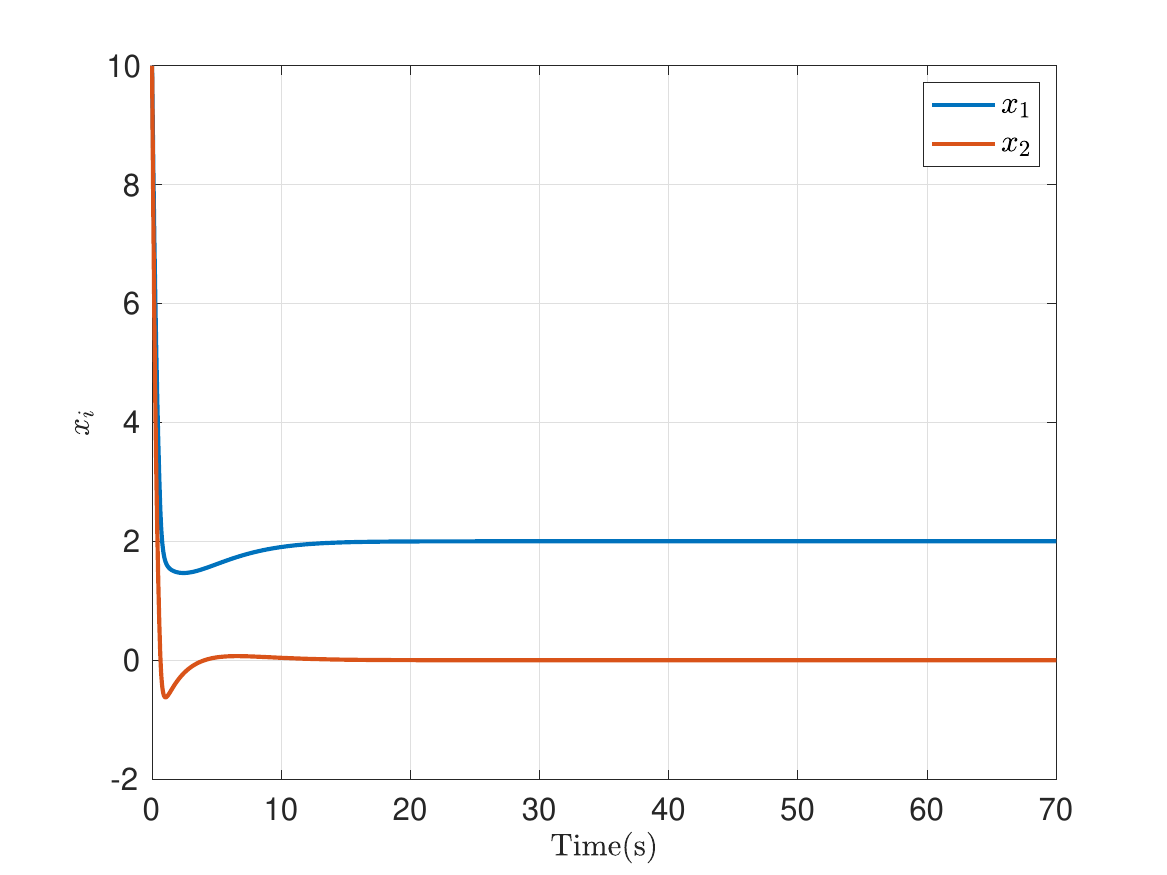}
\caption{Trajectories of $x$ over time.}
\label{fig:1}
\end{figure}
\begin{figure}[htbp]
\center
\includegraphics[width = 1\linewidth]{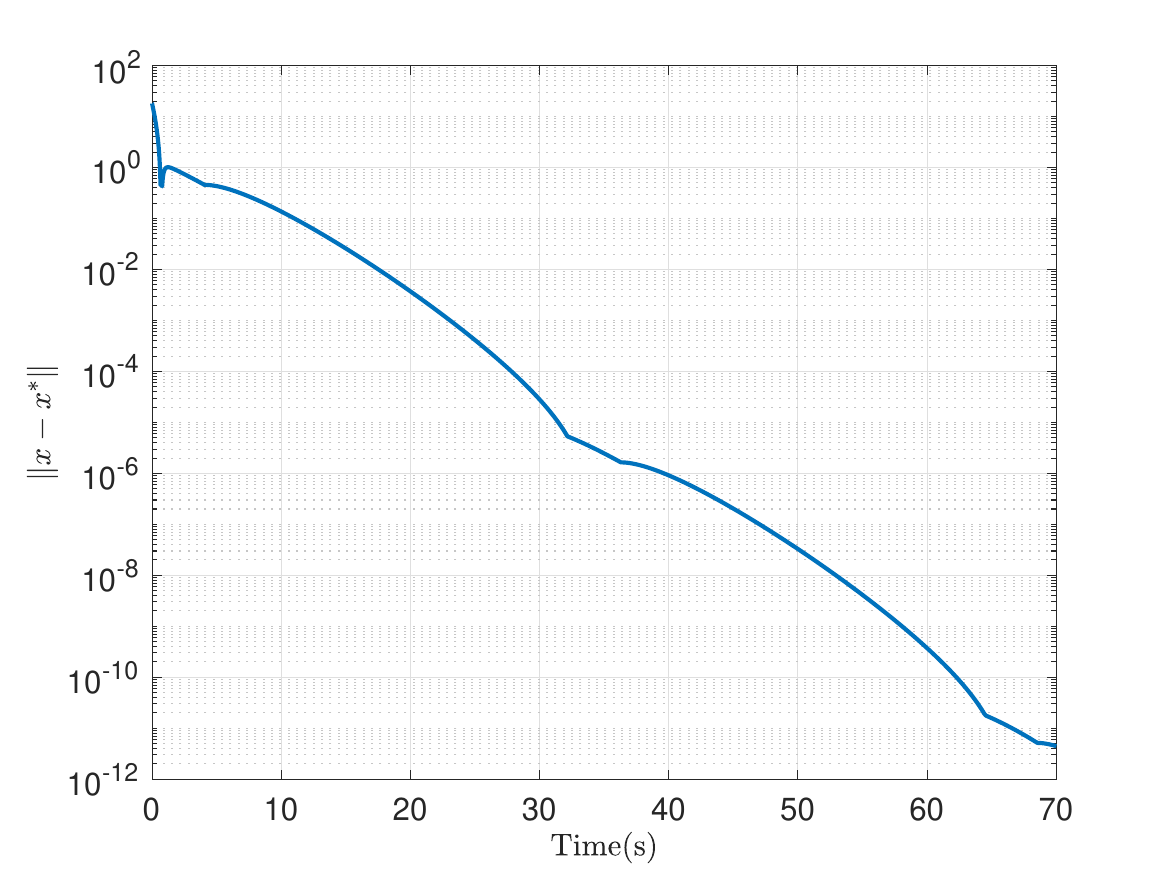}
\caption{Trajectory of the error $\| x - x^*\|$ over time.}
\label{fig:2}
\end{figure}
It can be observed that the augmented primal-dual gradient flow indeed achieves exponential convergence even without global strong convexity,  as suggested by Theorem~\ref{thm exponential convergence}.

As $ \nabla f(x) $ is piecewise continuous and the system matrices change depending on whether the inequality constraint is active or inactive,  the resulting dynamics yield the following three matrix inequalities,
\begin{align*}
\begin{bmatrix}
P H_i + H_i^\top P + 2 \alpha P & P B + \ell_i C^{\top} \\
B^\top P + \ell_i C & -2 I_2
\end{bmatrix}
\preceq 0, \quad i = 1, 2, 3,
\end{align*}
where \ml{$i = 1, 2, 3$ correspond to the regions $\{ x: x_1 + x_2 > 3 \}$, $\{ x: 2 \leq x_1 + x_2 \leq 3 \}$, and $\{ x: x_1 + x_2 < 2 \}$, respectively,} with $\ell_1 = 1$, $\ell_2 = \ell_3 = 2$, and 
\begin{align*}
& H_1 = 
\begin{bmatrix}
\ml{-1} & \ml{-1} & -1\\
\ml{-1} & \ml{-2} & -1\\
1 & 1 & 0
\end{bmatrix},~~
H_2= 
\begin{bmatrix}
-1.438 & -1 & -1\\
-1 & -1.438 & -1\\
1 & 1 & 0
\end{bmatrix},\\
& H_3 = 
\begin{bmatrix}
-0.438 & 0 & 0\\
0 & -0.438 & 0\\
0 & 0 & -1
\end{bmatrix}.
\end{align*}
By a bisection search, we obtain \ml{$\alpha = 0.123$}, and 
\begin{align*}
\ml{P = 
\begin{bmatrix}
2.317 & -0.457 & 0.138\\
\star & 2.216 & 0.406\\
\star  & \star & 2.310
\end{bmatrix} \succ 0}
\end{align*}
which certifies the feasibility of \eqref{eq: lmi search} in Corollary~\ref{cor}.

\section{Conclusion}\label{sec: conclusion}
We have shown that a common Lyapunov matrix exists if and only if the intersection of the set of strict Lyapunov matrices for one matrix and the set of non-strict Lyapunov matrices for the other is nonempty.
This result was then applied to the exponential analysis of the augmented primal-dual gradient flow for constrained optimization under affine inequality constraints.
Possible future extensions include investigating gradient flow under nonlinear inequality constraints as well as convergence properties for distributed algorithms.

\section*{Declaration of Generative AI and AI-assisted technologies in the writing process}
The author used ChatGPT solely for language polishing.
%


\bibliography{ifacconf}

\end{document}